\documentclass[conference]{IEEEtran}
\usepackage{graphicx}
\usepackage{algorithm2e}
\usepackage{algorithmic}
\usepackage{amsmath}
\usepackage{amsfonts}
\usepackage{amssymb}
\usepackage{wasysym}
\usepackage{multirow}
\newtheorem{thrm}{Theorem}
\newtheorem{proof}{Proof}
\begin{document}
\title{Secret Sharing Homomorphism and Secure E-voting}
\author{\IEEEauthorblockN{Binu V.P}
\IEEEauthorblockA{Department of Computer Applications\\
Cochin University\\
Kochi, India\\
binuvp@gmail.com}
\and
\IEEEauthorblockN{Divya G Nair}
\IEEEauthorblockA{Department of Computer Science \\
	Cochin University\\
	Kochi, India\\
	divyagnr@gmail.com}
\and
\IEEEauthorblockN{Sreekumar A}
\IEEEauthorblockA{Department of Computer Applications \\
Cochin University\\
Kochi, India\\
sreekumar@cusat.ac.in}
}
\maketitle
\begin{abstract}
Secure E-voting is a challenging protocol.Several approaches based on homomorphic crypto systems, mix-nets  blind signatures are proposed in the literature.But most of them need complicated homomorphic encryption which involves complicated encryption decryption process and key management which is not efficient.In this paper we propose a secure and efficient E-voting scheme based on secret sharing homomorphism.Here E-voting is viewed as special case of multi party computation where several voters jointly compute the result without revealing his vote.Secret sharing schemes are good alternative for secure multi party computation and are computationally efficient and secure compared with the cryptographic techniques.It is the first proposal, which  make use of the additive homomorphic property of the Shamir's secret sharing scheme and the encoding-decoding of votes to obtain the individual votes obtained by each candidates apart from the election result.We have achieved  integrity and privacy while keeping the efficiency of the system.
\end{abstract}
Keywords:E-voting, homomorphism,multi-party computation,secret sharing,individual votes
\section{Introduction}
% no \IEEEPARstart
Voting is a distributed decision making process involving several people.Each participant called the voter casts a vote and the computations are performed  on the vote casts by different voters to select the preferred item.Voting can be modeled as a secure multi party computation system, since multiple parties submit input and obtain the result without knowing any details of other inputs.

The  process involved in traditional election is quite tedious, time and resource consuming. To overcome these difficulties E-voting system  is introduced.The evolving new technologies made e-voting practical.But the research in this direction has to go a long way.The reliability and security are the major challenge.E-voting provides a lot of benefits compared with traditional voting.It avoids the requirement of geographical proximity of users.The cost can be greatly reduced because the resources can be reused. The use of E-voting must satisfy the security requirements such as authentication, voter privacy, confidentiality, integrity, etc.The security flaws make E-voting vulnerable than traditional system.Gritzalis et al \cite{gritzalis2002principles},\cite{gritzalis2003secure} mentioned the requirements of a secure E-voting system.

Confidentiality,Authenticity,Integrity and Verifiability are the major security requirements in E-voting scenario.Confidentiality ensures that nobody knows whom the voter is voted.Authentication is an important process where each voter must be identified as a person he claims to be and he should not be allowed to vote again.Integrity of the votes are also important.The system should ensure that the votes are valid and any modification must be detected.Verifiability means any one can verify at later time that, the voting is properly performed or his vote was properly registered and has been taken into account in the final tally \cite{fujioka1993practical}.

The are several proposal for efficient secret ballot elections based on mix-nets \cite{park1994efficient} \cite{sako1995receipt} \cite{jakobsson1998practical}, homomorphic encryption \cite{cohen1985robust} \cite{sako1994secure} \cite{benaloh1987verifiable} \cite{benaloh1994receipt} \cite{cramery1997secure} and blind signatures \cite{fujioka1993practical} \cite{okamoto1998receipt}.
There are different methods addressing the security and reliability of the E-voting scheme.Most of the approaches are based on cryptography.The major objective is to protect the voters identity from the vote.Secure E-voting using Blind Signature is proposed in \cite{ibrahim2003secure}.RSA \cite{rivest1978method} and Blind signatures are the major cryptographic algorithms involved\cite{camenisch1995blind},\cite{chaum1983blind}.Homomorphic encryption techniques are used in several implementations \cite{peng2005multiplicative}.Several proposal for verifiable secret sharing schemes are also given in \cite{benaloh1987verifiable},
\cite{benaloh1994receipt}.
Several modifications and use of homomorphic encryption and verifiable secret shuffle are mentioned in \cite{lee2000receipt}
\cite{hirt2000efficient} \cite{neff2001verifiable}.
Malkhi et al \cite{malkhi2003voting} in 2003 gave constructions without cryptographic technique which uses secret sharing homomorphism.Iftene \cite{iftene2007general} in 2007 proposed a general secret sharing for E-voting using Chinese remainder theorem . Pailliar's crypto system and its application to voting is proposed by Damagaard et al \cite{damgaard2010generalization} in 2010.Discrete logarithm problem and secret sharing are used by Chen et al \cite{chen2014secure} in 2014.Scheme with enhanced confidentiality and privacy is suggested by Pan et al \cite{pan2014enhanced} in 2014.

Secret sharing and many variations of its form an important primitive in several security protocols and applications.In the proposed method we make of Shamir's \cite{shamir1979} secret sharing techniques and its additive homomorphism property for efficient E-voting and vote tallying.Hence it avoids the complicated encryption decryption process and key management.The secret sharing schemes are originally proposed by Shamir \cite{shamir1979} and Blackley \cite{blakley1979} in 1979.
The motivation was to safeguard cryptographic keys.	The secret keys are stored at several locations as shares and when authorized number of users collaborate together, they can retrieve the secret.This provides both security,reliability and trust.The shares are information theoretically secure and provides no information about the key.The schemes are $(t,n)$ threshold schemes where any $t$ number of users can collaborate to recover the secret out of $n$ users.Less than $t$ users cannot obtain any information about the secret.Secret sharing scheme is perfect if less than $t$ shares gives no information about the secret.The secret sharing scheme is ideal if the share size is same as the secret size.Shamir's scheme is perfect and also ideal.It is easy to implement and is based on polynomial interpolation.Blackley's scheme is not perfect.It is based on hyperplane geometry and it is difficult to implement.There are other schemes based on boolean operations \cite{kuri2009xor} \cite{wang2007ssboln} and number theory \cite{mignotte1983} \cite{asmuth1983}.

Properties of polynomials give Shamir's scheme a $(+,+)$ homomorphic property.The secret domain and the share domain is same ( integers modulo $p$).There are other schemes \cite{asmuth1983modular},\cite{kothari1985generalized} which also have $(+,+)$ homomorphism, we consider Shamir's scheme for the ease of implementation and also it is perfect and information theoretically secure.
The homomorphism also provides verifiable secret sharing.It is very important in secure multi party computation.The first proposal of verifiable secret sharing was done by Chor et al \cite{chor1985verifiable}.In secret sharing not only the participant but also the dealer may be malicious.So the participant must be able to verify whether the shares are consistent.A set of $n$ shares is $t$ consistent if every subset of $t$ of the $n$ shares defines the same secret.Publicly verifiable secret sharing scheme's are introduced by Stadler in 1996 \cite{stadler1996publicly}.Schoenmakers \cite{schoenmakers1999simple} in 1999 proposed a publicly verifiable secret sharing scheme(PVSS) with applications to E-voting.The scheme is better than \cite{cramery1997secure},\cite{cohen1985robust}. The issue of homomorphic secret sharing for PVSS is also discussed.An efficient PVSS is suitable for the secure implementation of E-voting.

\section{Preliminaries}

The proposed scheme uses Shamir's threshold $(t,n)$ secret sharing scheme \cite{shamir1979}, which is having all the required property for efficient implementation of E-voting.The cryptographic techniques used for E-voting makes use of several mathematical assumptions.The use of secret sharing homomorphism avoids these problems and provides perfect secrecy.This section explores Shamir's secret sharing technique and the use of homomorphism property for e-voting.

\subsection{Shamir's Secret Sharing Scheme}

Let $n,t \in \mathbb{N}, t\le n$, a $(t,n)$ secret sharing protocol allows the shares of the secret to be distributed to $n$ participants and any $t$ of them can collaborate to retrieve the secret.Shamir \cite{shamir1979} uses a polynomial based construction for the implementation of  $(t,n)$ threshold scheme.It is based on the following theorem.

\begin{thrm}
	Let $l,t \in \mathbb{N}$.Also let $x_i,y_i \in \mathbb{Z}/p\mathbb{Z}, 1 \le i \le l$, where $p$ is prime and $x_i$ are pairwise distinct.Then there exist $p^{t-l}$ polynomial $q \in \mathbb{Z}/p\mathbb{Z}[X]$ of degree $\le t-1$
	with $q(x_i)=y_i,\; 1 \le i \le l$
\end{thrm}
\begin{proof}
	The polynomial can be obtained from the points $(x_i,y_i)$ using Lagrange interpolation formula.
	$$
	q(X)=\sum\limits_{i=1}^{t}y_i \prod\limits_{j=1,j\ne i}^{t}\frac{x_j-X}{x_j-x_i}$$
	
It satisfies $q(x_i)=y_i, 1 \le i \le l$.

We can determine number of such polynomials of degree $ \le t-1$.
$$q(X)= \sum\limits_{i=0}^{t-1} a_iX^i \; ,\; a_i \in \mathbb{Z}/p\mathbb{Z},\; 0\le i \le t-1$$ 

We can obtain a system of linear equations from $q(x_i)=y_i, 1 \le i \le l$
$$
 \begin{pmatrix} 1 & x_1 & x_1^2 &\cdots & x_1^{t-1} \\ 
                 1 & x_2 & x_2^2 &\cdots & x_2^{t-1} \\
           \vdots  & \vdots & \vdots &\vdots & \vdots \\
                 1 & x_l & x_l^2 &\cdots & x_l^{t-1} 
                  \end{pmatrix} \times 
                  \begin{pmatrix}
                 a_0 \\
                 a_1\\
                 \vdots \\
                 a_{t-1} 
                   \end{pmatrix}= \begin{pmatrix}
                   y_1 \\
                   y_2\\
                   \vdots \\
                   y_l 
                   \end{pmatrix}
$$
Lets consider the coefficient matrix 
$$C=\begin{pmatrix} 1 & x_1 & x_1^2 &\cdots & x_1^{l-1} \\ 
1 & x_2 & x_2^2 &\cdots & x_2^{l-1} \\
\vdots  & \vdots & \vdots &\vdots & \vdots \\
1 & x_l & x_l^2 &\cdots & x_l^{l-1} 
\end{pmatrix}
$$

The coefficient matrix is a Vandermonde matrix \cite{bjorck1970solution}. Since the $x_i's$ are distinct the determinant $\prod\limits_{i \le i \le j \le l}^{}(x_i-x_j)$ is non zero and the coefficient matrix has rank $l$.Thus the kernal of the coefficient matrix has rank $t-l$.So the number of polynomials is $p^{t-l}$ since each coefficient can take any of the $p$ values.
\end{proof}
\  \\
The secret sharing protocol consist of following phases.\\
\subsubsection{Initialization}
In this phase, the dealer who wants to distribute the secret $K$ must choose a prime number $p$ larger than the secret and also $p$ must be greater than $n$, where $n$ is the number of participants.The dealer then choose different $x_i \in \mathbb{Z}/p\mathbb{Z},\; 1 \le i \le n$ corresponds to each participants.The $x_i's$ are then published.\\
\subsubsection{Secret Sharing}
The secret $K \in \mathbb{Z}/p \mathbb{Z}$ is distributed as shares to the participants securely in this phase.

The dealer chooses $a_i \in \mathbb{Z}/p \mathbb{Z},\; 1 \le i \le t-1$ and construct a polynomial $q(X)$ of degree $t-1$ such that the constant term $q(0)$ represents the secret.
$$ q(X)=K + \sum_{i=1}^{t-1} a_iX^i$$

The shares corresponds to each participants are constructed by evaluating the polynomial at corresponding $x_i$ values.$$y_i=q(x_i), \; 1 \le i \le n$$.

The shares,  $y_i$ are then distributed to $i$'th participants
securely $1 \le i \le n$.

\subsubsection{Secret Reconstruction}

When $t$ or more participants collaborate together,they can retrieve the secret $K$ by combining the shares.Let $y_i=q(x_i),\; 1 \le i \le t$ ,be the $t$ shares.Then by using Lagrange Interpolation, the polynomial of degree $t-1$ can be reconstructed from these $t$ points using the formula 

$$ 	q(X)=\sum\limits_{i=1}^{t}y_i \prod\limits_{j=1,j\ne i}^{t}\frac{x_j-X}{x_j-x_i}$$

As per theorem 1 , there is exactly one such polynomial of degree $\le t-1$. The participants can obtain the secret $K$ as
$$K=q(0)=\sum\limits_{i=1}^{t}y_i \prod\limits_{j=1,j\ne i}^{t}\frac{x_j}{x_j-x_i}$$

It is noted that less than $t$ share holders  get no information about the secret.All constant term are equally likely and is an element in the field.The scheme is information theoretically secure.

\subsection{Secret Sharing Homomorphism}

Secret sharing homomorphism introduced by Benaloh in 1987 \cite{benaloh1987secret}.It is noted that Shamir's scheme is additive homomorphic.He stated that any $t$ of the $n$ agents can determine the super secret and no conspiracy of fewer than $t$ agents can gain any information at all about any of the sub secrets. That is the sum of the shares of different sub secret when added up and then interpolate according to the threshold mentioned to obtain the master secret which is the sum of the sub secrets.He also mentioned the importance of secret sharing homomorphism to E-voting.

Shamir's secret sharing scheme has the (+,+) homomorphism property. For example, assume there are two secrets: $K_1$, $K_2$ and are shared using polynomials $g(X)$ and $f(X)$.If we add the shares $h(i)=g(i)+f(i),\quad 1 \le i \le n$ then each of these $h(i)$ can be treated as the shares corresponds to the secret $K_1+K_2$.The polynomial $h(X)=g(X)+ f(X)$ and $h(0)=K_1+K_2$.
But each voter choose 1 or 0 ( vote or no vote).The shares are send to $n$ tellers.Any $t$ of them can collaborate to retrieve the result back.

In case of PVSS, two operations are defined.One on the shares  $\oplus$, and the other operation $\otimes$ on the encrypted shares such that for all participants
$$E_i(s_i) \otimes E_i(s_i')=E(s_i\oplus s_i')$$
If the underlying secret sharing scheme is homomorphic then by decrypting the combined encrypted shares, the recovered secret will be equal to $s_i \oplus s_i'$.

\section{Proposed Scheme}

The proposed system is a modification of the existing electronic voting scheme's used in India.Currently electronic voting machines are used in polling booth.These machines are costly and also not reliable.We propose an alternative solution for this using Internet and secret sharing homomorphism.This add trust and reliability to the existing voting scheme by incorporating secret sharing based techniques.The secrecy of vote is an important issue.This needs to be addressed with ultimate care. In the current Electronic Voting System, when a vote is casted,the corresponding candidates data base entry is updated and it can be easily tracked.But in the proposed scheme, it is difficult to track the vote because the shares of the votes is added to all the servers. We also add trust to the existing scheme by maintaining more than one server to keep the voting details.We are not considering on-line verification of the authenticity of the voter as in general e-voting scheme.Here we assume that the polling officers in each polling booth has to do it manually using the electoral role.The major components of the proposed e-voting schemes are
\begin{enumerate}
	\item Voting Terminal
	\item Share Generation 
	\item Collection Centers
	\item Result Computation
\end{enumerate}

We have considered the user authentication process which is done manually.The voting takes place in a Polling station.A voter is allowed to vote after his identity is verified. A polling station may contain many voting terminals.The user interface shows a voting panel which contains the list of all contesting candidates and their party symbols. Voting panel is setup and managed by the Chief Election Officers.

The share generation module is responsible for receiving the vote casted by each voter and make shares of it according to the threshold secret sharing scheme.
The shares are generated according to the vote casted for each candidates.Each candidate vote is represented as an encoded binary code.So when a vote is casted, the shares of the decimal value corresponds to the encoded binary vote of each candidate is generated using the Shamir's secret sharing scheme.The number of bits in the encoded binary code corresponds to each candidate vote depends on the number of contesting candidates and also total number of voters.

Let us assume that there are $m$ candidates $C_1,C_2,\ldots,C_m$ and $n$ voters $V_1,V_2,\ldots,V_n$.Then the binary encoding of the vote corresponds to each candidate will consist of $(\lfloor log_2n \rfloor+1)\times m$ number of bits.Here we consider the fact that all voters may vote to the same candidate.So the number of bits required for the representation of votes for each candidate is equal to the number of bits required to represent the total number of voters which is $\lfloor log_2n \rfloor+1$.

The encoding of the vote corresponds to each contesting candidate is explained below with an example.Let us consider that there are three candidates and seven voters.So the total number of bits of each encoded vote will be nine.The bit pattern corresponds to the vote of each candidate is obtained by setting the corresponding bit $C_i$ to 1 in the code $00C_300C_200C_1$ and other bits $C_i$ to 0.For example the code corresponds to the vote of candidate $C_3$ is $001000000 (64)$.So depending on the vote casted, it is encoded into a decimal code of 1,8 or 64 respectively.This bit wise encoding helps in computing  the total votes obtained by each candidate using the additive homomorphism.
 
The encoded vote is then shared using Shamir's threshold secret sharing scheme.The shares are then send to different Collection Centres(CC).The Collection Centres are responsible for receiving and summing up the shares corresponds to each vote casted.We can set up the threshold and also set number of collection centres required.If there are $p$ collection centres $CC_1,
CC_2,\ldots,CC_p$ and a threshold $t<p$ is set so that we can get back the result from any $t$ collection centres.This provides trust and reliability.Based on the number of collection centres and threshold set up, Shamir's scheme can be used for a threshold $(t,p)$ secret sharing.A $t-1$ degree polynomial $q(x)$ is constructed with constant term representing the encoded vote value in decimal.The other coefficients are chosen randomly from the field $\mathbb{Z}_p$, where $p$ is larger than the encoded vote values and the number of participants.The shares are generated by evaluating the polynomial $q(x)$ at $p$ different values $x_1,x_2,\ldots,x_p$.These $x_1,x_2,\ldots,x_p$ values represent different collection centres known only to the Chief Election Officer and are kept secret.These shares are then send securely to the $p$ collection centres.Any $t$ of them can be used for result evaluation and verification.The shares look totally random and the collection centres have no idea regarding which secret (vote) share it is, from the share value.The share size is also same as the secret size and hence it provides information theoretical security.
Once all the collection centres receives the vote share, the voting terminal is intimated to receive the next vote or it is the confirmation that the vote is registered properly.

The collection centres are responsible  for summing up the shares they receive for vote tallying.Here the shares are always valid. They are generated automatically from the terminal program embedded.So there is no need
to check the consistency of the shares received by the collection centres.But proper measures must be taken for the secure and error free communication between the voting terminal and collection centres.Collection centres behave as group of authorized parties.In a real time voting scenario, a single machine can act as a collection centre by maintaining database which contains collection of shares.How ever in this case the collection centre must be trusted.We can maintain a hierarchy of collection  centres for collecting vote shares according to the geographical location which compute the local sum of shares. The local sum is then send to the top level collection centres which further add the sums of shares received from local collection centres.A separate communication module can be incorporated for the efficient and secure communication of shares.The collection centres can also keep the shares received from each polling booth
or booths belongs to the same area as a separate entity for the computation of region wise voting details.The strategy for share maintenance, number of collection centres etc can be determined based on the requirement.The implementation issues also depend on the hierarchal structure used.

The Result Computation module is responsible for computing and declaring the final result.From the sum of shares stored on collection centres, the result can be obtained using Lagrange Interpolation.If there are $p$ collection centres and a $(t,p)$ threshold secret secret sharing scheme is used , then any $t$ of the share sum from the collection centres can be used for computing the final result.These $t$ shares can be used to get back a $t-1$ degree polynomial $Q(x)$ and the encoded result will be $Q(0)$.The result is then decoded by converting $Q(0)$ into binary and then separating the bits corresponds to each candidates.The decimal equivalent of the separated bits represent the total vote obtained by each candidate.Based on this the result can be announced.

It is noted that the result computation cannot be performed by collection centres.They will just keep the share sum and a hash is computed which is then signed by the private keys of the collection centre.During the result computation,it can be verified for the integrity and authenticity.This result declaration module, is managed by  higher officials and only they know the different $x$ values used for each collection centre during the share generation.Any $t$ pairs of this $x$ values and the corresponding share sum, which is the $y$ values, the polynomial interpolation can be done.The result computation can be done with different combination of the share sum from $t$ different collection centres which adds reliability.The trust is maintained by the Shamir's scheme because less than $t$ collection centres cannot get any information about the final result.At least $t$ collection centres have to collate to get back the result.

\section{Proposed Algorithms}
The following algorithm only includes the core functionality required.Additional functionalities can be added depending on the requirement.Suitable hash algorithm and signature algorithm must be chosen for maintaining the integrity and authenticity.When the voting is finished the hash of final share sum of each collection centre $SCC_j$ can be computed using SHA(Secure Hash Algorithm) \cite{pub2014draft} and is digitally signed by the previously issued private keys of the collection centre.The election official can verify this for integrity and authenticity.
\RestyleAlgo{boxruled}
\begin{algorithm}
	\begin{scriptsize}
		\KwIn{Vote casted by the voters}
		\KwOut {Sum of the shares of the votes}
		\BlankLine
		Let $m$ denote the number of candidates and 
		$n$ denote number of voters.\\
		Set $V$ equals $(\lfloor log_2n \rfloor+1)\times m$ bits
		for encoding the votes. \\
		Choose an appropriate field $\mathbb{Z}_p.$
	
		\For {each vote $i=1:n$  }{
			enc\textunderscore vote = bin\textunderscore decimal(set\textunderscore bit ($V$)) \\
			\tcc{$V$ is set according to the vote casted }
			\tcc{enc\_vote is the encoded vote in decimal }
			Pick $t-1$ random numbers $a_1,a_2,a_3,\ldots ,a_{t-1}$ from $\mathbb{Z}_p$\\
			Construct the polynomial \\ $q(x)$=enc\textunderscore vote $+a_1x+a_2x^2+\cdots+a_{t-1}x^{t-1}$\\
		\For {$j=1:p$}{
				Generate share $V_{ij}=q(j)$ \\
				\tcc{ where $V_{ij}$ is the $j^{th}$ share of $i^{th}$ vote}
				Send the share $V_{ij}$ to $C_{j}^{th}$ collection centre \\ through a secure communication channel\\
			 }
		\For {each Collection Centre $j=1:p$ }{
			 	Sum of shares $SCC_j=SCC_j+V_{ij}$ }
			}
	\caption{E-Voting}
	\label{Alg:Evt}
	\end{scriptsize}
	\end{algorithm}
	
	\begin{algorithm}
	\begin{scriptsize}
		\KwIn{Share sum of the votes from collection centre}
		\KwOut {Votes obtained by each candidate}
		\BlankLine
		\For {each randomly chosen $t$ Collection Centre $j=1:t$}  {
			retrieve $SCC_j$ } 
		Interpolate using $SCC_j$ and corresponding $x_i$ values to obtain the  polynomial Q(x)\\
		Obtain the secret value $Q(0)$.\\
		Decode $Q(0)$ and obtain the binary representation.\\
		Each $(\lfloor log_2n \rfloor+1)$ bits will represent each candidates vote.\\
		Publish the final results.
		\BlankLine
		\BlankLine
		\caption{Result Computation}
		\label{Alg:result}
	\end{scriptsize}
\end{algorithm}	
\section{E-voting Example}
Let us assume that three people Alice,Bob and Charles are contesting in an election and there are seven voters.So the maximum vote each contestant can get is seven.Three bits are hence required for the representation of votes gained by each candidate and a total of nine bits for the representation of encoded votes corresponds to each candidate.

$m=3, n=7, V=9 bits$\\

A sample voting scenario is given below where six voters made the vote out of seven.
\begin{table}[ht]
	\small
		\caption{Example E-voting}
	% title of Table
	\centering
\begin{tabular}{|c|c|c|c|c|c|} \hline
 \multirow{2}{*}{Vote} & \multirow{2}{*} {Alice} &\multirow{2}{*}  {Bob} & \multirow{2}{*}{Charles} & \multicolumn{2}{c|}{Encoding Vote}\\
\cline{5-6}
    &   & & & Binary & Decimal \\
    \hline
	1 & \checked & 		   &  & 001000000 & 64 \\ 
	\hline 
	 2& 		    & \checked &  & 000001000 & 8 \\ 
	\hline 
	 3&  		& \checked &  &000001000   &  8\\ 
	\hline  
	4&\checked &  		   &  & 001000000 & 64 \\ 
	\hline 
	 5&  		&			&\checked  & 000000001 & 1 \\ 
	\hline
	 6 &\checked  &  &  & 001000000 & 64 \\ 
	\hline 
\end{tabular} 
\end{table}

The votes are encoded as shown in Table I corresponds to each contesting candidate.Lets choose a field $\mathbb{Z}_{257}$.We have considered a $(2,3)$ secret sharing scheme where any two shares can be used to reconstruct the secret .Every time a vote is casted, a random polynomial $q(x)$ of degree 1 are constructed with constant term as the encoded vote and the other coefficient are
 chosen randomly from $\mathbb{Z}_257$.Generate the three shares $CC_1,CC_2$ and $CC_3$ with $x_i$'s as 1,2 and 3.It is noted that the shares are random irrespective of the encoded vote.So the collection centre cannot derive any information about the secret from the shares they receive.The collection centre also compute the share sum $SCC_j$ from the shares they receive.Table II shows the random polynomials constructed,the corresponding shares generated and also the share sum in the sample run of the algorithm corresponds to $(2,3)$ secret sharing scheme.
\begin{table}[ht]
	\small
	\caption{Vote Sharing}
	% title of Table
	\centering
\begin{tabular}{|c|c|l|c|c|c|} \hline
	\multirow{2}{*}{Vote} & \multirow{2}{*} {enc \textunderscore vote} & \multirow{2}{*}  {q(x)} &  \multicolumn{3}{c|}{Shares}\\
	\cline{4-6}
   & & & $CC_1$ & $CC_2$ & $CC_3$ \\
	\hline
	1 & 64 & 233.x+64 & 40 & 16 & 249\\ 
	\hline 
  	2 & 8 & 157.x+8 & 165 & 65 & 222 \\ 
	\hline 
	3 & 8 & 78.x+8 & 86 & 164 & 242 \\ 
	\hline  
	4 & 64 & 255.x+64 & 62 & 60 & 58 \\ 
	\hline 
	5 & 1 & 217.x+1 & 218 & 178 & 138 \\ 
	\hline
	6 & 64 & 124.x+64& 188 & 55 & 179\\ 
	\hline 
	\multicolumn{3}{|c|}{share sum $SCC_j$} & 245 & 24 & 60 \\
	\hline
\end{tabular} 
\end{table}

The election result can be computed from the sum of shares $SCC_j$ maintained by  each collection centre using Lagrange interpolation.The polynomial
$Q(x)$ can be obtained using any two shares in the example using the Lagrange Interpolation formula as follows.
$$Q(x)= SCC_1. \frac{(x-x_2)}{(x_1-x_2)} + SCC_2. \frac{(x-x_1)}{(x_2-x_1)}$$

The final result depends on $Q(0)$ which is easily obtained by
$$Q(0) = SCC_1. \frac{(x_2)}{(x_2-x_1)} + SCC_2. \frac{(x_1)}{(x_1-x_2)}$$

Computation of results using different combination of shares $SCC_1:SCC_2,SCC_1:SCC_3 \;\mbox{and}\; SCC_2:SCC_3$ are shown in equation 1,2 and 3.The operations are carried out in $\mathbb{Z}_{257}$.It is noted that the reconstructed values are consistent.
\begin{eqnarray}
Q(0) &=  245.\frac{2}{2-1} + 24. \frac{1}{1-2} =&  209 \\
Q(0) &=  245.\frac{3}{3-1} + 60. \frac{1}{1-3} =&  209 \\
Q(0) &=  24.\frac{3}{3-2} + 60. \frac{2}{2-3}  =&  209 
\end{eqnarray}

The final result can be obtained by decoding the reconstructed result 209 into
binary.It is noted that 3 bits will represent vote secured by each candidate.

$$\mbox{decoded vote}:\; 011, 010, 001$$

The result can be published based on the obtained values which is shown below.

\begin{table}[ht]
	\small
	\caption{E-voting Result}
	% title of Table
	\centering
	\begin{tabular}{|c|c|} \hline
		 \textbf{Candidate}	& \textbf{Votes Secured} \\
		\hline
		Alice & 3 \\
		\hline
		Bob &   2 \\
		\hline
		Charles & 1 \\
		\hline
	\end{tabular} 
\end{table}

\section{Analysis}

Security in online election is a challenging task.Authenticating the voter is a major challenge along with the privacy of the vote.We have considered manual authentication and proposed a modification to the existing voting scheme which uses electronic voting machine.The voting machines are not reliable and also in certain situations where the number of candidates are more, more than one voting machine needs to be connected.The proposed scheme is cost effective and also reliable.

It is noted that the proposed algorithm mentioned is simple and effective and provides privacy to the vote casted.The shares are generated  by constructing a random polynomial and the share size is same as the encoded vote.The collection centers have no idea about how the votes are encoded, how many bits are used for encoding, which bits represents a particular candidate votes etc.The collection centres will receive a random value from the field $\mathbb{Z}_p$ from which no information about the secret vote can be obtained.The coalition of $t$ untrusted collection centers can obtain the result.But they doesn't have any knowledge about the number of collection centres,the threshold used and also what is the $x$ values assigned to each collection centre.In the example we have considered 1,2 and 3 for simplicity, however different $x$ values can be used and is kept secret.

Shamir's secret sharing scheme is information theoretically secure.It is perfect in the sense that no information can be obtained from less than the threshold number of participants.This adds trust to the existing E-voting scheme,because the computation of the result need participation of $t$ centers.The computation of the shares and the reconstruction of the final result using the share sum can be done using polynomial evaluation and interpolation.Efficient $O(n.log^2n)$ algorithms for polynomial evaluation and interpolation are mentioned in \cite{aho1974design},\cite{lloyd1982art}.Simple quadratic algorithms are sufficient because the number of shares generated is not too large.

The encoding and decoding of the votes can also be done easily.The codes for each candidates and also the number of bits required to represent the votes depends on the number of voters and number of contesting candidate.These setups are done by the election officials prior  to the election process.The decoding of votes is a simple binary conversion which can also be done easily.
The integrity of the share sum maintained by each centre is achieved by implementing a digital signature scheme.This can also be efficiently implemented
using any digital signature scheme \cite{atreya2002digital}. 

The algorithm is computationally efficient and the complexity involved  depends on the share generation during the voting and the communication with the Collection Centres .The number of shares are usually small and hence the share generation using polynomial evaluation is simple.The secure communication between the voting terminal and the collection centre is a more challenging.Separate communication module can be incorporated to do it efficiently.The collection centre must also be capable of handling requests from large number of voting terminals.Region wise collection centres can be incorporated to balance the load and update the top level collection centre data in a periodic manner.The result analysis needs the polynomial interpolation but is done only once and it doesnt add much complexity to the performance of the system.

\section{Conclusions}
\label{section:conclusions}
The E-voting scheme using Shamir's secret sharing homeomorphism is a first proposal which helps to obtain not only the election result but also the votes gained by each candidate with encoding and decoding of votes in a typical manner.The proposal of secret sharing homomorphism was suggested by several authors, however true or false voting mechanism is mentioned.The proposed algorithm generalize the use of secret sharing homomorphism to E-voting which provides secrecy, computational efficiency, trust and reliability.The system does not also leave any trace of the vote made by a voter.

The strong requirement of the scheme mentioned here is a secure channel for sending shares.The shares can be send through different channels to different collection centres.The intruder have to get access to $t$ different channels for breaking the security of the scheme.For additional security, the shares can also be encrypted by using the public keys of the collection centre.There are several homomorphic encryptions which support this or ordinary encryption decryption can be used.

The system works efficiently for a moderate election with less number of voters.
If the number of voters and candidates are more, the encoded vote will have a large value and the system has to chose a field of large size.This will large share size and too much communication overhead.This can be avoided by breaking the encoded vote into smaller code and makes shares of it.However the complexity involved in the implementation will increase.

We have done a preliminary implementation of the scheme using java \cite{nair2015improved}.Additional modules are incorporated as per the requirement.Another feature that can be incorporated is the implementation of digital signature scheme,which ensures integrity and authenticity of the shares.Verifiable secret sharing techniques can also be incorporated which ensures the consistency of the shares, however it slow down the system performance. We are looking for a more sophisticated implementation guaranteeing authentication using mobile phones and OTP(One Time Password) for all the users using adhar details.Instead of voting terminals every one can vote using the registered mobile which is our future plan.

\end{document}